\documentclass[nofootinbib,onecolumn,11pt]{revtex4}
\usepackage{amsmath}
\usepackage{amsthm}
\usepackage{amsfonts}
\usepackage{graphicx}
\usepackage{amssymb}
\usepackage{color}
\usepackage{rotating}

\usepackage{latexsym}
\newtheorem{theorem}{Theorem}

\newtheorem{lemma}[theorem]{Lemma}

\def\ll{\lambda}

\begin{document}
\title{On moments of the integrated exponential Brownian motion }
\author{F. Caravelli$\ ^{1,2,3}$, T. Mansour$\ ^4$, L. Sindoni$\ ^5$, S. Severini$\ ^3$\\\ \\
{\small $\ ^1$ Invenia Labs, 135 Innovation Dr., Winnipeg, MB R3T 6A8, Canada\\
$\ ^2$ London Institute for Mathematical Sciences, 35a South Street, London W1K 2XF, UK\\
$\ ^3$Department of Computer Science, University College London, Gower Street, London WC1E 6BT, UK \\
$\ ^4$ Department of Mathematics, University of Haifa, 31905 Haifa, Israel\\
$\ ^5$Max Planck Institute for Gravitational Physics, Albert Einstein Institute,\\Am M\"{u}hlenberg 1, 14467 Golm, Germany
}}
\begin{abstract}
We present new exact expressions for a class of moments for the geometric Brownian motion, in terms of determinants, obtained using a recurrence relation and combinatorial arguments for the case of a Ito's Wiener process. We then apply the obtained exact formulas to computing averages of the solution of the logistic stochastic differential equation via
a series expansion, and compare the results to the solution obtained via Monte Carlo.
\end{abstract}
\maketitle
\section{Introduction}
The geometric Brownian motion is the stochastic process described by the differential equation
\begin{equation}
df=\mu f dt+\sigma f dW_t,
\end{equation}
where $W_t$ is a Wiener process and $\mu, \sigma$ are constants describing the drift and the variance of the noise, respectively. The solution can be written as
\begin{equation}
f(W_t,t)=\exp\left\{\left(\mu-\frac{\sigma^2}{2}\right)t+\sigma W_t\right\}.
\end{equation}
Geometric Brownian motion is used for modelling many phenomena in a variety of contexts \cite{Gardiner}. A prominent role is played in financial applications, where the distribution of returns can be approximated by a log-normal distribution \cite{YorBook,Gardiner}, at least in specific regimes. 

For the computation of certain properties, it is necessary to compute the integral of $f(W,t)$
over a time interval
\begin{equation}
F[W,t]= \int_0^t f(W_s,s) ds.
\end{equation}
The evaluation of this functional is also involved in the solution of the geometric Brownian motion with logistic corrections. In general, averages of the form
\begin{equation}
\langle  G(F[W,t])  \rangle = \sum_{k=0}^\infty a_k \langle F[W,t]^k\rangle \equiv \sum_{k=0}^\infty a_k r_k.
\end{equation}
are quite common.
The evaluation of averages of powers of the integrated exponential Brownian motion, then, is
instrumental for the computation of these observables. 

Detailed studies of this functional and of its powers \cite{Rev1,Rev2} are already available in the literature.

In this paper, we will derive exact formulas for the evaluation of these integrals, under the assumption of the Ito formulation for the Wiener process. Similar results have been given in \cite{Yor1,YorBook}. Motivated by obtaining exact formulas for Asian options, in \cite{Yor1} Yor obtained an exact formula in terms of polynomials for the following moments:
\begin{equation}
\left\langle e^{\sigma W_t} \left(\int_0^t d \tilde t\ e^{W_{\tilde t}}\right)^n\right\rangle.
\end{equation}
Using Girsanov's theorem \cite{Girsanov}, one can derive a series of identities, in which the last is Bougerol's formula
\begin{equation}
\left\langle e^{\sigma W_t} \left(\int_0^t d\tilde t e^{W_{\tilde t}}\right)^n\right\rangle=\left\langle P_n\left(e^{2 B_t}\right)\right\rangle=4^n \left\langle \frac{\sinh( W_t)^{2n}}{\langle W_1^{2n}\rangle } \right\rangle,
\end{equation}
where
\begin{equation}
P_n(x)=\Gamma(n) \sum_{j=0}^n c_j z^j,
\end{equation}
and
\begin{equation}
c_j=\prod_{k\neq j\\0\leq k\leq n} \frac{2}{(\mu+j)^2-(\mu+k)^2}.
\end{equation}

In this work, we take a different route with the use of combinatorics. 
We prove a recurrence relation for the integrals involved at the $k$-th order in 
terms of integrals at the $(k-1)$-th order, and after resummation, we get an identity in terms of a 
determinant.
\section{Calculation of moments}
The central quantity of interest in the present paper is given by the average over the Wiener process $W_s$:
\begin{equation}
r_k(\mu,\sigma,t)\equiv\langle F[W,t] ^k \rangle.
\label{eq:cumulant}
\end{equation}
If we expand Eq. (\ref{eq:cumulant}), we obtain
\begin{equation}
\langle F[W,t] ^k \rangle=\int_0^t d {\tilde t }_k \cdots \int_0^t d {\tilde t }_1 \left\langle e^{\sum_{i=1}^k [(\mu-\frac{\sigma^2}{2}){\tilde t}_i+\sigma W_{{\tilde t}_i}]} \right\rangle.
\end{equation}
We will use the following formula due to the properties of integrals with Gaussian measure \cite{Salmhofer,Gardiner}, and in which we assume that $\langle W_{t} W_{t^\prime} \rangle$ is of the Ito type. This implies
\begin{eqnarray}
\left\langle e^{\sigma \sum_{i=1}^k W_{{\tilde t}_i}} \right\rangle &=&  e^{\frac{\sigma^2}{2} \sum_{i,j=1}^k \left\langle W_{{\tilde t}_i} W_{{\tilde t}_j}
\right\rangle} \nonumber \\
&=&  e^{\frac{\sigma^2}{2} \sum_{i,j=1}^k \text{min}({\tilde t}_i,{\tilde t}_j)}.
\label{eq:ito}
\end{eqnarray}
By using this property, we can now prove the following fact:
\begin{lemma}\label{lem1}
For the average over the Wiener process $W_s$ of Ito type, the following formula holds true:
\begin{equation}
r_k(\mu,\sigma,t)=\Gamma(n) \int_0^t e^{\mu {\tilde t}_k}\int_0^{{\tilde t}_{k-1}} e^{(\mu+ \sigma^2) {\tilde t}_{k-2}}\cdots \int_0^{{\tilde t}_{3}} e^{ (\mu+(k-2) \sigma^2) {\tilde t}_2} \int_0^{{\tilde t}_{2}} e^{ (\mu+(k-1) \sigma^2) {\tilde t}_1} d{\tilde t}_1 \cdots d{\tilde t}_k
\end{equation}
\end{lemma}
\begin{proof}
By a direct application of Eq. (\ref{eq:ito})
\begin{eqnarray}
r_{k}(\mu,\sigma,t)&=&\left\langle \underbrace{\int_0^t \cdots \int_0^t}_{k} e^{\sum_{i=1}^k[(\mu-\frac{\sigma^2}{2})  {\tilde t}_i+ \sigma W_{{\tilde t}_i}]} d{\tilde t}_1 \cdots d{\tilde t}_k \right\rangle \nonumber \\
&=&\int_0^t \cdots \int_0^t e^{\sum_{i=1}^k[(\mu-\frac{\sigma^2}{2})  {\tilde t}_i+ \frac{\sigma^2}{2} \sum_j \text{min}({\tilde t}_i,{\tilde t}_j)} dt_1 \cdots dt_k.
\end{eqnarray}
Due to symmetry of integrand, we can order the integration variables as\ ${\tilde t}_i<{\tilde t}_{i+1}$, obtaining
\begin{eqnarray}
r_{k}(\mu,\sigma,t)&=&  \Gamma(k) \int_0^t \int_0^{{\tilde t}_{k-1}} \cdots \int_0^{{\tilde t}_{2}} e^{\sum_{i=1}^k \mu  {\tilde t}_i+ \frac{\sigma^2}{2} \sum_{i\neq j; i,j=1}^k \text{min}({\tilde t}_i,{\tilde t}_j)]} d{\tilde t}_1 \cdots d{\tilde t}_k ,
\end{eqnarray}
whence:
\begin{eqnarray}
r_{k}(\mu,\sigma,t)&=&\Gamma(k) \int_0^t \int_0^{{\tilde t}_{k-1}} \cdots \int_0^{{\tilde t}_{2}} e^{\sum_{i=1}^k \mu  {\tilde t}_i+ \sigma^2 \sum_{i<j}^n \text{min}({\tilde t}_i,{\tilde t}_j)]} d{\tilde t}_1 \cdots d{\tilde t}_k  \nonumber \\
&=&\Gamma(k) \int_0^t \int_0^{{\tilde t}_{k-1}} \cdots \int_0^{{\tilde t}_{2}} e^{\sum_{i=1}^k \mu  {\tilde t}_i+ \sigma^2 \sum_{i=1}^k (k-i){\tilde t}_i]} d{\tilde t}_1 \cdots d{\tilde t}_k.
\end{eqnarray}
After rearranging carefully the terms, we arrive at the final result:
\begin{eqnarray}
r_{k}(\mu,\sigma,t) &=&\Gamma(k) \int_0^t e^{\mu {\tilde t}_k}\int_0^{{\tilde t}_{k-1}} e^{(\mu+ \sigma^2) {\tilde t}_{k-2}}\cdots \int_0^{{\tilde t}_{3}} e^{ (\mu+(k-2) \sigma^2) {\tilde t}_2} \int_0^{{\tilde t}_{2}} e^{ (\mu+(k-1) \sigma^2) {\tilde t}_1} d{\tilde t}_1 \cdots d{\tilde t}_k  \nonumber \\
\label{eq:formula}
\end{eqnarray}
\end{proof}
Let us now expand further on Eq. (\ref{eq:formula}). It is convenient to first perform the rescaling
$t_i=t u_i$. Then, by defining $\lambda_j=t (\mu+(k-j)\sigma^2)$, we obtain
\begin{eqnarray}
r_k(\vec \lambda,t)&=& t^k \Gamma(k) \int_0^1 du_k \int_0^{u_k}d u_{k-1}\cdots \int_0^{u_2}du_1 e^{ \sum_{i=1}^k \lambda_i u_i} \equiv
t^k \Gamma(k) s_k({ \lambda}_1,\cdots,{ \lambda}_k)   
\end{eqnarray}
Therefore, the computation of \eqref{eq:cumulant} reduces to the computation of
\begin{equation}
s_k({ \lambda}_1,\cdots,{ \lambda}_k)=\int_0^1 du_k \int_0^{u_k}d u_{k-1}\cdots \int_0^{u_2}du_1 e^{ \sum_{i=1}^k  \lambda_i u_i}.
\end{equation}
a very similar formula had been obtained in \cite{Yor1}; the main difference between the treatment made there and our relies on what follows. An important observation will enable us to evaluate these integrals exactly by means of combinatorics: it is possible in fact
to prove the following result, which establishes a recursion relation among the $s_k$.
\begin{lemma}
For the quantity $s_k({ \lambda}_1,\cdots,{ \lambda}_k)$, we have
\begin{equation}
s_k({ \lambda}_1,\cdots,{ \lambda}_k)=\frac{e^{{ \lambda}_k}s_{k-1}({ \lambda}_1,\cdots,{ \lambda}_{k-1})-s_{k-1}({ \lambda}_1,\cdots{ \lambda}_{k-2},{ \lambda}_{k-1}+{ \lambda}_{k})}{{ \lambda}_k},
\label{eq:recursion}
\end{equation}
with $s_0=1$.
\end{lemma}
\begin{proof}
Let us write
\begin{eqnarray}
s_k({ \lambda}_1,\cdots,{ \lambda}_k)&=&\int_0^1 du_k \int_0^{u_k}d u_{k-1}\cdots \int_0^{u_2}du_1 e^{ \sum_{i=1}^k  \lambda_i u_i}\nonumber \\
&=&\int_0^1 du_k e^{ \lambda_k u_k} f(u_k),
\end{eqnarray}
with $f(u_k)=\int_0^{u_k}d u_{k-1}\cdots \int_0^{u_2}du_1 e^{ \sum_{i=1}^{k-1}  \lambda_i u_i}$.
Integrating by parts,
\begin{equation}
s_k({ \lambda}_1,\cdots,{ \lambda}_k)=\left.\frac{e^{{ \lambda}_k u_k}}{{ \lambda}_k} f(u_k)\right|^1_0-\int_0^1 d u_k \frac{e^{{ \lambda}_k u_k}}{{ \lambda}_k} f'(u_k),
\end{equation}
where
\begin{eqnarray}
f'(u_k)&=&\int_0^{u_{k-1}} \cdots \int_0^{u_2}du_1 e^{ \sum_{i=1}^{k-2}  \lambda_i u_i+ \lambda_{k-1} u_k} \nonumber \\
&=& e^{\lambda_{k-1} u_k}\int_0^{u_{k-1}} \cdots \int_0^{u_2}du_1 e^{ \sum_{i=1}^{k-2}  \lambda_i u_i}
\end{eqnarray}
Finally, the identity $f(0)=0$ gives the desired result.
\end{proof}
Eq. (\ref{eq:recursion}) suggests the evaluation of the averages by means of combinatorial considerations. 
Indeed, the evaluation of the integral can proceed graphically, for any fixed order of the moment $k$, as in Fig. \ref{fig:rec}. Starting from the top and using the properties of the recurrence relation, at each order $k$, we organize the recurrence on a binary tree. In Fig. \ref{fig:rec}, each left branch will pull out a factor $e^{\lambda}/{\lambda}$, where $\lambda$ is the $\lambda$ obtained from the previous order. One has to consider the fact that however, in each right branch one pulls out a factor $-1/\lambda$, and sums the two factors of $\lambda$'s. To obtain the final formula, once the empty set has been reached, one multiplies the final term by all the factors in the branch.
\begin{figure}
\centering
\includegraphics[scale=0.4]{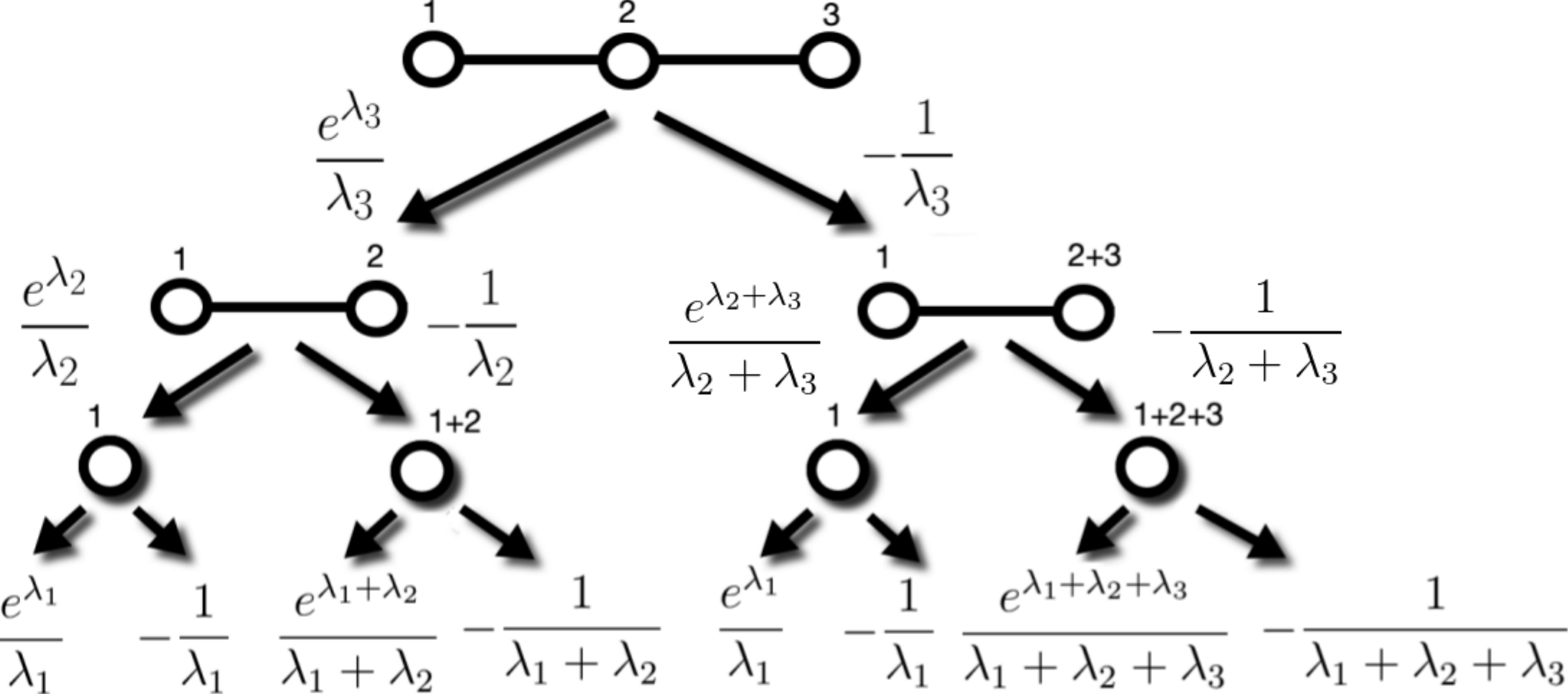}
\caption{Application of the graphical method for evaluating $k$-th moments of the exponential integrated Gaussian process to the case $k=3$.}
\label{fig:rec}
\end{figure}

We now give exact formulas for all the terms obtained from the recurrence. To lighten the 
notation,
we define $\mu_{k,j}=\ll_j+\cdots+\ll_k$ and $s_k=s_k(\ll_1,\ldots,\ll_k)$.

By iterating the second term of the recurrence, we obtain
\begin{align*}
s_k&=\frac{e^{\mu_{k,k}}}{\mu_{k,k}}s_{k-1}-\frac{e^{\mu_{k,k-1}}}{\mu_{k,k}\mu_{k,k-1}}s_{k-2}+\frac{1}{\mu_{k,k}\mu_{k,k-1}}
s_{k-2}(\ll_1,\ldots,\ll_{k-3},\mu_{k,k-2})\\
&=\frac{e^{\mu_{k,k}}}{\mu_{k,k}}s_{k-1}-\frac{e^{\mu_{k,k-1}}}{\mu_{k,k}\mu_{k,k-1}}s_{k-2}+\frac{e^{\mu_{k,k-2}}}{\mu_{k,k}\mu_{k,k-1}\mu_{k,k-2}}s_{k-3}
-\frac{1}{\mu_{k,k}\mu_{k,k-1}\mu_{k,k-2}}s_{k-3}(\ll_1,\ldots,\ll_{k-4},\mu_{k,k-3})\\
&=\cdots\\
&=\sum_{j=1}^{k}(-1)^{k-j}\frac{e^{\mu_{k,j}}}{\mu_{k,j}\mu_{k,j+1}\cdots\mu_{k,k}}s_{j-1}+(-1)^k\frac{1}{\mu_{k,1}\cdots\mu_{k,k}}.
\end{align*}
Therefore,
\begin{align}
s_k(\ll_1,\ldots,\ll_k)&=\sum_{j=1}^{k}(-1)^{k-j}\frac{e^{\ll_j+\cdots+\ll_k}}{\prod_{i=j}^k(\ll_i+\cdots+\ll_k)}s_{j-1}(\ll_1,\ldots,\ll_{j-1})+\frac{(-1)^k}{\prod_{i=1}^k(\ll_i+\cdots+\ll_k)}.\label{eqss1}
\end{align}

\begin{lemma}\label{lem2}
Let $f_k=\sum_{j=0}^ke_j^kf_{j-1}$ with $f_{-1}=1$. Then
$$f_k=\sum_{m=0}^k\sum_{0=i_0<i_1<\cdots<i_m<k+1=i_{m+1}}e_{i_m}^{i_{m+1}-1}e_{i_{m-1}}^{i_m-1}\cdots e_{i_0}^{i_1-1}.$$
\end{lemma}
\begin{proof}
We proceed the proof by induction on $k$. We have $f_k=e_0^0$, which holds for $k=0$. We assume that the claim holds for $k$ and let us prove it for $k+1$. By the recurrence relation, we have
\begin{align*}
f_{k+1}=e_0^{k+1}+\sum_{j=1}^{k+1}e_j^{k+1}f_{j-1}.
\end{align*}
By induction hypothesis, we have
\begin{align*}
f_{k+1}&=e_0^{k+1}+\sum_{j=1}^{k+1}e_j^{k+1}\left(\sum_{m=0}^{j-1}\sum_{0=i_0<i_1<\cdots<i_m<j=i_{m+1}}e_{i_m}^{j-1}e_{i_{m-1}}^{i_m-1}\cdots e_{i_0}^{i_1-1}\right)\\
&=e_0^{k+1}+\sum_{j=1}^{k+1}\left(\sum_{m=0}^{j-1}\sum_{0=i_0<i_1<\cdots<i_m<j=i_{m+1}<i_{m+2}=k+2}e_{i_{m+1}}^{i_{m+2}-1}e_{i_m}^{j-1}e_{i_{m-1}}^{i_m-1}\cdots e_{i_0}^{i_1-1}\right)\\
&=e_0^{k+1}+\sum_{m=1}^{k+1}\left(\sum_{j=m}^{k+1}\sum_{0=i_0<i_1<\cdots<i_{m-1}<j=i_{m}<i_{m+1}=k+2}e_{i_{m}}^{i_{m+1}-1}e_{i_{m-1}}^{i_m-1}\cdots e_{i_0}^{i_1-1}\right)\\
&=e_0^{k+1}+\sum_{m=1}^{k+1}\sum_{0=i_0<i_1<\cdots<i_{m}<i_{m+1}=k+2}e_{i_{m}}^{i_{m+1}-1}e_{i_{m-1}}^{i_m-1}\cdots e_{i_0}^{i_1-1}\\
&=\sum_{m=0}^{k+1}\sum_{0=i_0<i_1<\cdots<i_{m}<i_{m+1}=k+2}e_{i_{m}}^{i_{m+1}-1}e_{i_{m-1}}^{i_m-1}\cdots e_{i_0}^{i_1-1},
\end{align*}
which completes the induction step.
\end{proof}

If $e_j^k=(-1)^{k-j}\frac{e^{\ll_j+\cdots+\ll_k}}{\prod_{i=j}^k(\ll_i+\cdots+\ll_k)}$ and $e_0^k=\frac{(-1)^k}{\prod_{i=1}^k(\ll_i+\cdots+\ll_k)}$, then Eq. \eqref{eqss1} can be written as
$$s_k(\ll_1,\ldots,\ll_k)=\sum_{j=1}^{k}e_j^ks_{j-1}(\ll_1,\ldots,\ll_{j-1})+e_0^k.$$
Applying Lemma \ref{lem2}, we obtain
\begin{align*}
&s_k(\ll_1,\ldots,\ll_k)=\\
&=\sum_{m=0}^k\sum_{0=i_0<i_1<\cdots<i_m<k+1=i_{m+1}}e_{i_m}^{i_{m+1}-1}e_{i_{m-1}}^{i_m-1}\cdots e_{i_0}^{i_1-1}\\
&=\sum_{m=0}^k\sum_{0=i_0<i_1<\cdots<i_m<k+1=i_{m+1}}\frac{(-1)^{i_1-1}}{\prod_{j=1}^{i_1-1}(\ll_j+\cdots+\ll_{i_1-1})}
\prod_{\ell=1}^m\frac{(-1)^{i_{\ell+1}-i_\ell-1}e^{\ll_{i_\ell}+\cdots+\ll_{i_{\ell+1}-1}}}{\prod_{j=i_\ell}^{i_{\ell+1}-1}(\ll_j+\cdots+\ll_{i_{\ell+1}-1})}\\
&=\sum_{m=0}^k\sum_{0=i_0<i_1<\cdots<i_m<k+1=i_{m+1}}\frac{(-1)^{k-m}e^{\ll_{i_1}+\cdots+\ll_k}}
{\prod_{j=1}^{i_1-1}(\ll_j+\cdots+\ll_{i_1-1})\prod_{\ell=1}^m\prod_{j=i_\ell}^{i_{\ell+1}-1}(\ll_j+\cdots+\ll_{i_{\ell+1}-1})},
\end{align*}
which leads to the following result.
\begin{theorem}
For all $k\geq0$, $s_k(\ll_1,\ldots,\ll_k)$ is given by
\begin{align*}
\sum_{m=0}^k\sum_{0=i_0<i_1<\cdots<i_m<k+1=i_{m+1}}\frac{(-1)^{k-m}e^{\ll_{i_1}+\cdots+\ll_k}}
{\prod_{j=1}^{i_1-1}(\ll_j+\cdots+\ll_{i_1-1})\prod_{\ell=1}^m\prod_{j=i_\ell}^{i_{\ell+1}-1}(\ll_j+\cdots+\ll_{i_{\ell+1}-1})}.
\end{align*}
\end{theorem}
\medskip

The above theorem gives
\begin{align*}
&s_k(t\ll_1,\ldots,t\ll_k)=\\
&\sum_{m=0}^k\sum_{0=i_0<i_1<\cdots<i_m<k+1=i_{m+1}}\frac{(-1)^{k-m}e^{t(\ll_{i_1}+\cdots+\ll_k)}}
{t^{k-1}\prod_{j=1}^{i_1-1}(\ll_j+\cdots+\ll_{i_1-1})\prod_{\ell=1}^m\prod_{j=i_\ell}^{i_{\ell+1}-1}(\ll_j+\cdots+\ll_{i_{\ell+1}-1})}.
\end{align*}
Hence, the coefficient of $t^n$ in $s_k(t\ll_1,\ldots,t\ll_k)$ is given by
\begin{align*}
&\frac{1}{(n+k-1)!}\sum_{m=0}^k\sum_{0=i_0<i_1<\cdots<i_m<k+1=i_{m+1}}\frac{(-1)^{k-m}(\ll_{i_1}+\cdots+\ll_k)^{n+k-1}}
{\prod_{j=1}^{i_1-1}(\ll_j+\cdots+\ll_{i_1-1})\prod_{\ell=1}^m\prod_{j=i_\ell}^{i_{\ell+1}-1}(\ll_j+\cdots+\ll_{i_{\ell+1}-1})}.
\end{align*}

A cleaner formula can be obtained in terms of determinants.  If we define $s_k=s_k(\ll_1,\ldots,\ll_k)$ for all $k\geq0$, where we define $s_0=s_{-1}=1$. If $e_j^k=(-1)^{k-j}\frac{e^{\ll_j+\cdots+\ll_k}}{\prod_{i=j}^k(\ll_i+\cdots+\ll_k)}$ and $e_0^k=\frac{(-1)^k}{\prod_{i=1}^k(\ll_i+\cdots+\ll_k)}$, then Eq. \eqref{eqss1} can be written as
$$s_k=\sum_{j=0}^{k}e_j^ks_{j-1}.$$
Using Theorem 4.20 in \cite{Vein} (see also \cite{Dets}), we obtain the following result:
\begin{theorem}
For all $k\geq0$, $s_k=s_k(\ll_1,\ldots,\ll_k)$ is given by
\begin{align}
\det\left(\begin{array}{lllllll}
e_k^k&e_{k-1}^k&e_{k-2}^k&\cdots&e_2^k&e_1^k&e_0^k\\
-1&e_{k-1}^{k-1}&e_{k-2}^{k-1}&\cdots&e_2^{k-1}&e_1^{k-1}&e_0^{k-1}\\
0&-1&e_{k-2}^{k-2}&\cdots&e_2^{k-2}&e_1^{k-2}&e_0^{k-2}\\
\vdots&\vdots&\vdots&&\vdots&\vdots&\vdots\\
0&0&0&&-1&e_1^1&e_0^1\\
0&0&0&&0&-1&e_0^0\\
\end{array}\right),
\label{eq:det}
\end{align}
where $e_j^k=(-1)^{k-j}\frac{e^{\ll_j+\cdots+\ll_k}}{\prod_{i=j}^k(\ll_i+\cdots+\ll_k)}$ and $e_0^k=\frac{(-1)^k}{\prod_{i=1}^k(\ll_i+\cdots+\ll_k)}$.
\end{theorem}
This expression can be elucidated by means of an example. Consider the case $n=3$ and $\sigma=0$. In this case, we have an exact formula for $\lambda_j=\mu t$:
$$n=3:\ {\it Det} \left(  \left[ \begin {array}{cccc} {\frac {{e^{\mu\,t}}}{
\mu\,t}}&-1/2\,{\frac {{e^{2\,\mu\,t}}}{{\mu}^{2}{t}^{2}}}&1/6\,{
\frac {{e^{3\,\mu\,t}}}{{\mu}^{3}{t}^{3}}}&-1/6\,{\frac {1}{{\mu}^{3}{
t}^{3}}}\\\noalign{\medskip}-1&{\frac {{e^{\mu\,t}}}{\mu\,t}}&-1/2\,{
\frac {{e^{2\,\mu\,t}}}{{\mu}^{2}{t}^{2}}}&1/2\,{\frac {1}{{\mu}^{2}{t
}^{2}}}\\\noalign{\medskip}0&-1&{\frac {{e^{\mu\,t}}}{\mu\,t}}&-{
\frac {1}{\mu\,t}}\\\noalign{\medskip}0&0&-1&1\end {array} \right]
 \right) =1/6\,{\frac { \left( {e^{\mu\,t}}-1 \right) ^{3}}{{\mu}^{3}{
t}^{3}}},$$
which provides the exact value obtained from the deterministic logistic equation.


The previous results are general enough to hold also for averages of the form
\begin{equation}
\tilde r_n(\mu,\sigma,t)=\langle e^{(\mu-\frac{\sigma^2}{2})t+\sigma W_t}F[W,t]^n\rangle=e^{(\mu-\frac{\sigma^2}{2})t} \langle e^{\sigma W_t}F[W,t]^n\rangle.
\end{equation}
It is easy to see that Eq. (\ref{eq:det}) applies. We just substitute ${\tilde \lambda}_k=\lambda_k+t \sigma^2$ and multiply  by a factor:
\begin{eqnarray}
\tilde r_n(\mu,\sigma,t)&=&e^{(\mu-\frac{\sigma^2}{2}) t} t^n e^{\frac{\sigma^2}{2} t} \Gamma(n) s_n(\tilde \lambda_1,\cdots, \tilde \lambda_n) \nonumber \\
&=&e^{\mu t} t^n  \Gamma(n) s_n(\tilde \lambda_1,\cdots, \tilde \lambda_n).
\label{eq:detgen}
\end{eqnarray}
The above expression will be used in the following section as an application to averages in the logistic stochastic differential equation.  Appendix \ref{sec:appendix} contains the analytical values of the functions $s_k(\lambda_1,\cdots,\lambda_k)$ and $s_k(\tilde \lambda_1,\cdots,\tilde \lambda_k)$ for the cases $n \leq 4$. It is immediate that, when $t\approx0$, we have
\begin{equation}
s_n(t \mu,t (\mu+\sigma^2),\cdots, t (\mu+(n-1)\sigma^2))\approx \frac{1}{t^n}.
\end{equation}
 This facts implies that Eq. (\ref{eq:detgen}) is analytic in $t=0$. Fig. \ref{fig:plotsn} is a plot of $s_n(t)$ for $\mu=1$, $\sigma=0.1$.

\begin{figure}
\includegraphics[scale=0.6]{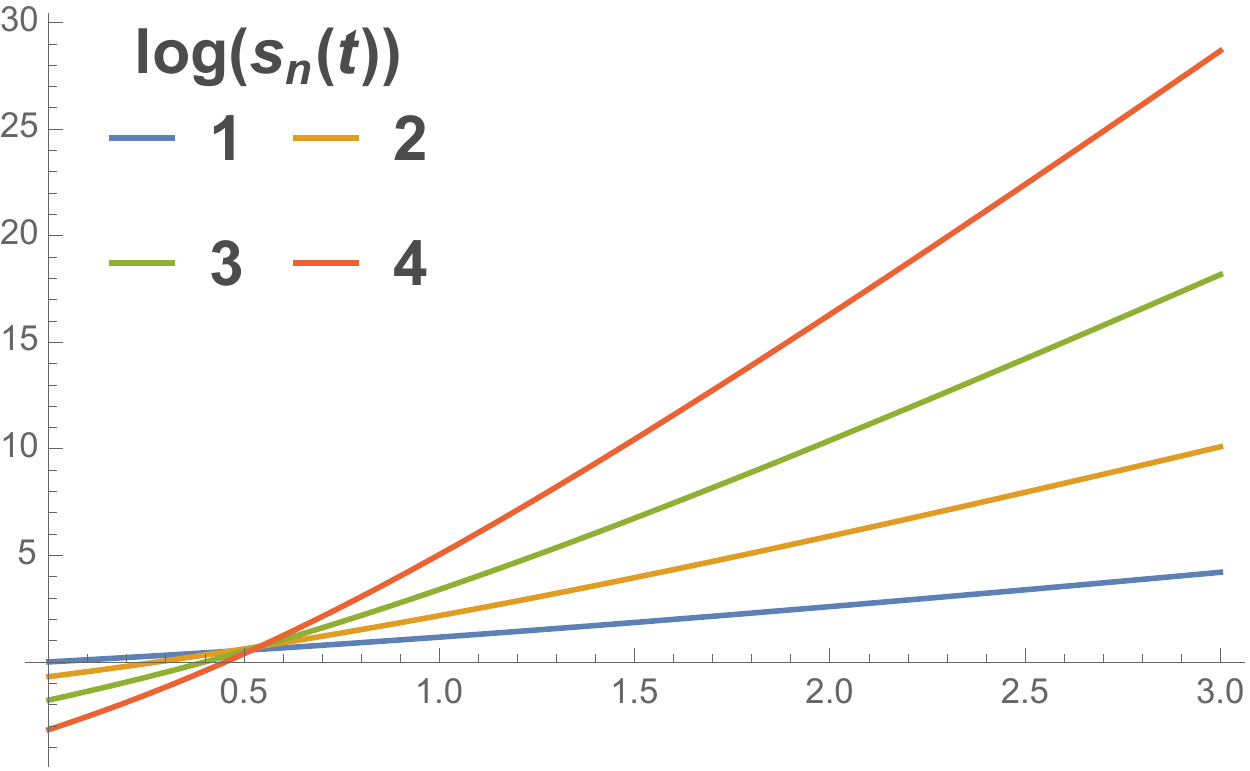}\includegraphics[scale=0.6]{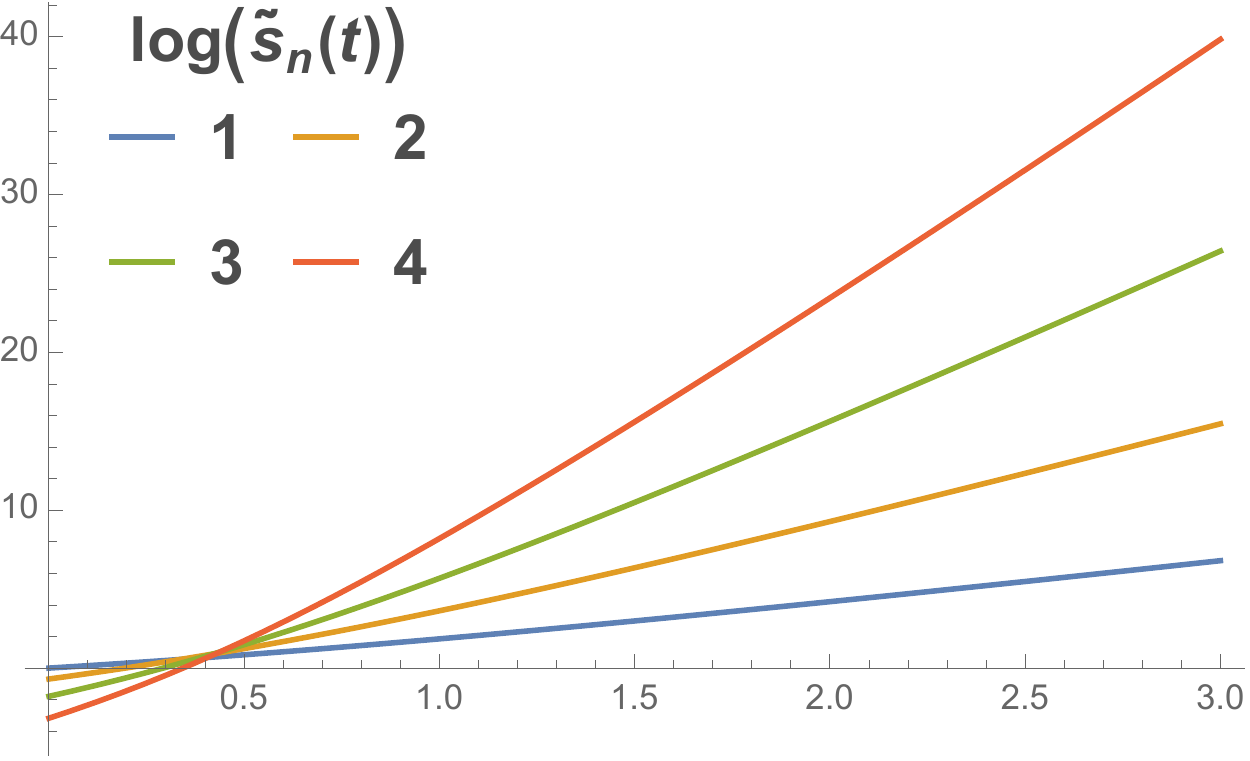}
\caption{Plot of the functions $s_n(\mu,\sigma,t)$ and $\tilde s_n(\mu,\sigma,t)$ provided in Appendix \ref{sec:appendix} as a function of $t$ for $\mu=3$, $\sigma=1$ for $n=1,2,3,4$.}
\label{fig:plotsn}
\end{figure}

\section{Averages of Logistic SDE in perturbation theory}
As an application of the formula (\ref{eq:detgen}), we focus on the solution of the logistic stochastic differential equation motion,
\begin{equation}
dx = x\left[\mu  \left(1- \frac{x}{\tilde x}\right) dt + \sigma dW\right],
\label{eq:diffeq}
\end{equation}
given by (we follow \cite{KloedenPlaten,CaravelliEtAl})
\begin{eqnarray}
x(t)&=&x_0\ e^{ (\mu-\frac{\sigma^2}{2})t+\sigma W_t }\left(1+\frac{\mu x_0}{\tilde x}\int_0^t e^{ (\mu-\frac{\sigma^2}{2})s+\sigma W_s } ds\right)^{-1} \nonumber \\
&= &x_0\ e^{ (\mu-\frac{\sigma^2}{2})t+\sigma W_t } \sum_{n=0}^\infty (-1)^n \left(\frac{\mu x_0}{\tilde x}\int_0^t e^{ (\mu-\frac{\sigma^2}{2})s+\sigma W_s } ds\right)^n.
\end{eqnarray}
We evaluate the average of the solution $x(t)$,
\begin{eqnarray}
\langle x(t)\rangle&= &\left\langle x_0\ e^{ (\mu-\frac{\sigma^2}{2})t+\sigma W_t } \sum_{n=0}^\infty (-1)^n \left(\frac{\mu x_0}{\tilde x}\int_0^t e^{ (\mu-\frac{\sigma^2}{2})s+\sigma W_s } ds\right)^n
\right\rangle \nonumber \\
&= & x_0\ \sum_{n=0}^\infty (-1)^n \left(\frac{\mu x_0}{\tilde x}\right)^n \left\langle e^{ (\mu-\frac{\sigma^2}{2})t+\sigma W_t } \left(\int_0^t e^{ (\mu-\frac{\sigma^2}{2})s+\sigma W_s } ds\right)^n
\right\rangle \nonumber \\
&= & x_0\ e^{\mu t} \sum_{n=0}^\infty (-1)^n \left(t \frac{\mu x_0}{\tilde x}\right)^n  \Gamma(n) s_n(\tilde \lambda_1,\cdots, \tilde \lambda_n).
\label{eq:pertth}
\end{eqnarray}
and observe that it involves the moments of (\ref{eq:detgen}). In the limit $\tilde x\rightarrow \infty$, Eq. (\ref{eq:pertth})
reduces to the average of the geometric Brownian motion. We consider now truncations of the mean of $x(t)$ at $k$-th order,
$\langle x(t) \rangle_k= x_0\ e^{\mu t} \sum_{n=0}^k (-1)^n (t \frac{\mu x_0}{\tilde x})^n  \Gamma(n) s_n(\tilde \lambda_1,\cdots, \tilde \lambda_n)$, and compare the truncated solution to the one obtained numerically. In Fig. \ref{fig:simulations}, we plot $\langle x(t) \rangle_k$ for $k=0,1,\ldots,4$  obtained for $\mu=1$, $\sigma=0.1$, $\tilde x=100$ by means of a stochastic Euler method with $dt=10^{-3}$, and the averages then obtained using Monte Carlo over $1000$ samples. We observe that the higher the order of the approximation, the closer we are to the solution obtained by Monte Carlo.

\begin{figure}
\includegraphics[scale=0.5]{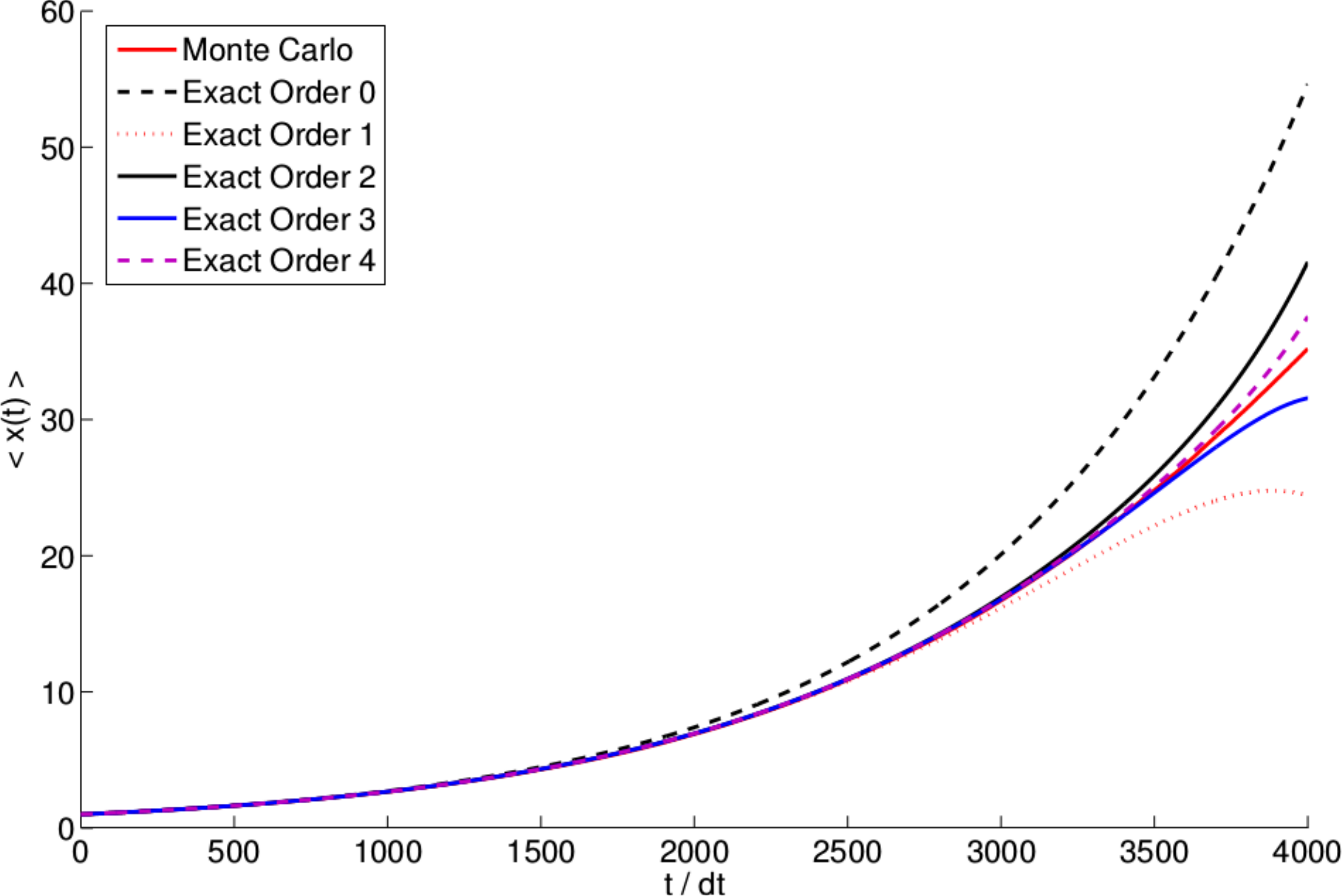}
\caption{Plot of the mean of solution of the logistic stochastic differential equations for $\mu=1$, $\sigma=0.1$, $\tilde x=100$, and with $dt=10^{-3}$, solved numerically using a stochastic Euler method and averaged over 1000 simulations (solid red), versus the analytical solution obtained at $k$-th order, considering $\xi=\frac{x_0}{\tilde x}$ the perturbative parameter. We can observe that at higher order we obtain a solution closer to the one simulated.}
\label{fig:simulations}
\end{figure}

\section{Conclusions}
In this paper, we have presented new exact formulas for the moments of the integrated exponential Brownian motion in terms of sums and determinants, and based on recent results obtained in \cite{Dets}. We described a simple graphical method to evaluate them, based on a recurrence relation. 

Exact formulas were proved in \cite{Yor1} in terms of polynomials. In this paper however, we have taken an alternative route based on combinatorics. After realizing that the mean can be evaluated exactly using the properties of Gaussian integrals, and after observing that these moments feature a recurrence relation, we have shown that exact expressions can be obtained via a combinatorial argument. These exact expressions were then observe to be equivalent to evaluating the determinant of a specific linear operator which depends on the order of the moment.

To complete the presentation, we have applied the formulas to the exact solution of the logistic 
stochastic differential equation. There, the evaluation of the ensemble expectation values
of certain observables can be carried out with our method, via Taylor expansion.
In particular, the comparison of the mean solution obtained by means of Monte Carlo simulations 
\footnote{The numerical integration of the differential equation is performed by a stochastic Euler method.} with our method shows that our formula permits to approximate to a higher precision 
some ensemble averages of properties of the solution of the stochastic differential equation.

\section*{Aknowledgements}
F. C. would like to thank the London Institute for Mathematical Sciences for hospitality. S. S. is supported by the Royal Society and EPSRC. 

\newpage
\appendix
\section{Exact formulas for $n\leq 4$} \label{sec:appendix}
\begin{sideways}
\parbox{\textheight}{
Below we report the functions $s_k(\mu,\sigma,t)\equiv s_k(\lambda_1,\cdots,\lambda_k)$ and $\tilde s_k(\mu,\sigma,t)\equiv s_k(\tilde \lambda_1,\cdots,\tilde \lambda_k)$ up to $k=4$.\\
\begin{eqnarray}
s_1(\mu,\sigma,t)&=&\frac{e^{\mu  t}-1}{\mu  t}; \ \ s_2(\mu,\sigma,t)=\frac{\frac{e^{t \left(2 \mu +\sigma ^2\right)}-1}{2 \mu +\sigma ^2}+\frac{1-e^{\mu
   t}}{\mu }}{t^2 \left(\mu +\sigma ^2\right)}; \\
 s_3(\mu,\sigma,t)&=& \frac{2 \mu ^2 \left(e^{3 t \left(\mu +\sigma ^2\right)}-3 e^{t \left(2 \mu +\sigma
   ^2\right)}+3 e^{\mu  t}-1\right)+6 \sigma ^4 \left(e^{\mu  t}-1\right)+\mu  \sigma ^2
   \left(e^{3 t \left(\mu +\sigma ^2\right)}-9 e^{t \left(2 \mu +\sigma ^2\right)}+15
   e^{\mu  t}-7\right)}{3 \mu  t^3 \left(\mu +\sigma ^2\right) \left(2 \mu +\sigma
   ^2\right) \left(\mu +2 \sigma ^2\right) \left(2 \mu +3 \sigma ^2\right)};\\
s_4(\mu,\sigma,t)&=&\frac{3 \mu ^2 \sigma ^2
   \left(-8 e^{3 t \left(\mu +\sigma ^2\right)}+18 e^{t \left(2 \mu +\sigma
   ^2\right)}+e^{4 \mu  t+6 \sigma ^2 t}-16 e^{\mu  t}+5\right)-30 \sigma ^6 \left(e^{\mu
    t}-1\right)}{6   \mu  t^4 \left(\mu +\sigma ^2\right) \left(2 \mu +\sigma ^2\right) \left(\mu +2 \sigma
   ^2\right) \left(\mu +3 \sigma ^2\right) \left(2 \mu +3 \sigma ^2\right) \left(2 \mu +5
   \sigma ^2\right)}  \nonumber \\
&+&\frac{2 \mu ^3 \left(-4 e^{3 t \left(\mu +\sigma ^2\right)}+6 e^{t \left(2 \mu +\sigma
   ^2\right)}+e^{4 \mu  t+6 \sigma ^2 t}-4 e^{\mu  t}+1\right)  }{6
   \mu  t^4 \left(\mu +\sigma ^2\right) \left(2 \mu +\sigma ^2\right) \left(\mu +2 \sigma
   ^2\right) \left(\mu +3 \sigma ^2\right) \left(2 \mu +3 \sigma ^2\right) \left(2 \mu +5
   \sigma ^2\right)}  \nonumber \\
   &+& \frac{\mu  \sigma ^4 \left(-10 e^{3 t \left(\mu +\sigma ^2\right)}+54 e^{t
   \left(2 \mu +\sigma ^2\right)}+e^{4 \mu  t+6 \sigma ^2 t}-82 e^{\mu  t}+37\right)}{6
   \mu  t^4 \left(\mu +\sigma ^2\right) \left(2 \mu +\sigma ^2\right) \left(\mu +2 \sigma
   ^2\right) \left(\mu +3 \sigma ^2\right) \left(2 \mu +3 \sigma ^2\right) \left(2 \mu +5
   \sigma ^2\right)};     \\
\tilde s_1(\mu,\sigma,t)&=&\frac{e^{t \left(\mu +\sigma ^2\right)}-1}{t \left(\mu +\sigma ^2\right)}; \ \ \tilde s_2(\mu,\sigma,t)= \frac{\frac{e^{2 \mu  t+3 \sigma ^2 t}-1}{2 \mu +3 \sigma ^2}-\frac{e^{t \left(\mu
   +\sigma ^2\right)}-1}{\mu +\sigma ^2}}{t^2 \left(\mu +2 \sigma ^2\right)}; \\
\tilde s_3(\mu,\sigma,t)&=& \frac{-\frac{\mu +3 \sigma ^2}{2 \mu ^2+5 \mu  \sigma ^2+3 \sigma ^4}+\left(\frac{2}{\mu
   +\sigma ^2}+\frac{1}{-2 \mu -5 \sigma ^2}\right) e^{t \left(\mu +\sigma
   ^2\right)}-\frac{3 e^{2 \mu  t+3 \sigma ^2 t}}{2 \mu +3 \sigma ^2}+\frac{e^{3 t
   \left(\mu +2 \sigma ^2\right)}}{2 \mu +5 \sigma ^2}}{3 t^3 \left(\mu +2 \sigma
   ^2\right) \left(\mu +3 \sigma ^2\right)} \\
   \tilde s_4(\mu,\sigma,t)&=& \frac{\frac{2}{2 \mu ^2+9 \mu  \sigma ^2+10 \sigma ^4}-\frac{1}{4 \mu ^2+24 \mu  \sigma
   ^2+35 \sigma ^4}-\frac{6 \left(\mu +4 \sigma ^2\right)}{4 \mu ^3+28 \mu ^2 \sigma
   ^2+61 \mu  \sigma ^4+42 \sigma ^6}+\frac{2 \left(\mu +4 \sigma ^2\right)}{\left(\mu
   +\sigma ^2\right) \left(\mu +2 \sigma ^2\right) \left(2 \mu +5 \sigma ^2\right)}}{6 t^4 \left(\mu +3 \sigma ^2\right) \left(\mu +4 \sigma
   ^2\right)} \nonumber \\
   &+&\frac{ e^{t
   \left(\mu +\sigma ^2\right)} \left(-\frac{2 \left(\mu +4 \sigma ^2\right)}{\left(\mu
   +\sigma ^2\right) \left(\mu +2 \sigma ^2\right) \left(2 \mu +5 \sigma
   ^2\right)}-\frac{2 e^{t \left(2 \mu +5 \sigma ^2\right)}}{2 \mu ^2+9 \mu  \sigma ^2+10
   \sigma ^4}+\frac{e^{3 t \left(\mu +3 \sigma ^2\right)}}{4 \mu ^2+24 \mu  \sigma ^2+35
   \sigma ^4}+\frac{6 \left(\mu +4 \sigma ^2\right) e^{t \left(\mu +2 \sigma
   ^2\right)}}{\left(\mu +2 \sigma ^2\right) \left(2 \mu +3 \sigma ^2\right) \left(2 \mu
   +7 \sigma ^2\right)}\right)}{6 t^4 \left(\mu +3 \sigma ^2\right) \left(\mu +4 \sigma
   ^2\right)};
\end{eqnarray}
}
\end{sideways}
\end{document}